\documentclass[11pt]{article}
\usepackage[
	includefoot,
	margin=3cm,
	a4paper,
]{geometry}
\usepackage{microtype}
\usepackage[utf8]{inputenc}
\usepackage{enumerate}
\usepackage{amsmath}
\usepackage{amsfonts}
\usepackage{nicefrac}
\usepackage{amssymb}
\usepackage{amsthm}
\usepackage{mathtools}
\usepackage{tikz}
\usetikzlibrary{calc}
\tikzstyle{para}=[rectangle,draw=black,minimum height=.8cm,fill=gray!10,rounded corners=1mm, on grid]
\usepackage{pgfplots}
\pgfplotsset{compat=newest}
\usepackage{pgfplotstable}
\usepackage{xcolor}
\usepackage{authblk}
\usepackage{thmtools, thm-restate}
\usepackage{dsfont}

\tikzstyle{path} = [color=black,opacity=.1,line cap=round, line join=round, line width=12pt]
\tikzset{
	colornode/.style = {
		circle,
		draw=#1!70!black,
		very thick,
		fill=#1
	}
}

\definecolor{r}{rgb}{1.0, 0.4, 0.4}
\definecolor{r0}{rgb}{1.0, 0.7, 0.4}
\definecolor{r1}{rgb}{1.0, 0.4, 0.0}
\definecolor{r2}{rgb}{0.8, 0, 0.4}
\definecolor{r3}{rgb}{1, 0.4, 0.3}
\definecolor{b}{rgb}{0.4, 0.4, 1.0}
\definecolor{b0}{rgb}{0.4, 0.7, 1.0}
\definecolor{b1}{rgb}{0.0, 0.4, 1.0}
\definecolor{b2}{rgb}{0.4, 0.0, 0.8}
\definecolor{g}{rgb}{0.2, 0.9, 0.3}
\definecolor{y2}{rgb}{1, 1, 0.3}

\usepackage{tabularx}
\usepackage{booktabs}
\usepackage{paralist}

\usepackage{algorithm2e}
\DontPrintSemicolon

\usepackage[pdfdisplaydoctitle,menucolor=orange!40!black,filecolor=magenta!40!black,urlcolor=blue!40!black,linkcolor=red!40!black,citecolor=green!40!black,colorlinks]{hyperref}

\newcommand{\ExternalLink}{%
    \tikz[color=magenta, x=1.2ex, y=1.2ex, baseline=-0.05ex]{%
        \begin{scope}[x=1ex, y=1ex]
            \clip (-0.1,-0.1) 
                --++ (-0, 1.2) 
                --++ (0.6, 0) 
                --++ (0, -0.6) 
                --++ (0.6, 0) 
                --++ (0, -1);
            \path[draw, 
                line width = 0.5, 
                rounded corners=0.5] 
                (0,0) rectangle (1,1);
        \end{scope}
        \path[draw, line width = 0.5] (0.5, 0.5) 
            -- (1, 1);
        \path[draw, line width = 0.5] (0.6, 1) 
            -- (1, 1) -- (1, 0.6);
        }
}

\DeclareUnicodeCharacter{1EF3}{\`y}
\usepackage[
	backend=biber,
	isbn=false,
	url=false,
	doi=true,
	backref=true,
	hyperref=true,
	natbib=true,
	maxcitenames=3,
	maxbibnames=99,
	sortcites]{biblatex}
\renewbibmacro{in:}{\ifentrytype{article}{}{\printtext{\bibstring{in}\intitlepunct}}}
\newbibmacro{doilink}[1]{%
	\iffieldundef{doi}{%
		\iffieldundef{url}{%
			#1%
        }{%
		\href{\thefield{url}}{$\ExternalLink$}%
		}%
	}{%
		\href{https://doi.org/\thefield{doi}}{$\ExternalLink$}%
	}%
}
\DeclareFieldFormat*{doi}{\usebibmacro{doilink}{#1}}

\bibliography{strings-long,bibliography}

\usepackage[nameinlink,sort&compress,capitalize,noabbrev]{cleveref}

\usepackage[backgroundcolor=gray!10,textsize=footnotesize]{todonotes}

\newtheorem{theorem}{Theorem}
\newtheorem{lemma}[theorem]{Lemma}

\newtheorem{corollary}[theorem]{Corollary}

\theoremstyle{definition}
\newtheorem{definition}[theorem]{Definition}

\newtheorem{alg}{Algorithm}

\crefname{rrule}{Rule}{Rules}
\crefname{figure}{Figure}{Figures}
\crefname{alg}{Algorithm}{Algorithms}

\newcommand{\prob}[1]{\textnormal{\textsc{#1}}}

\newcommand{\myproblem}[5]{
	\begin{center}
	\begin{minipage}{0.95\columnwidth}
		\noindent
		\prob{#1}
		\vspace{5pt}\\
		\setlength{\tabcolsep}{3pt}
		\begin{tabularx}{\textwidth}{@{}lX@{}}
			\textbf{#2}     & #3 \\
			\textbf{#4}  & #5
		\end{tabularx}
	\end{minipage}
	\end{center}
}

\newcommand{\problemdef}[3]{\myproblem{#1}{Input:}{#2}{Question:}{#3}}

\DeclarePairedDelimiterX{\abs}[1]{\lvert}{\rvert}{#1}
\DeclarePairedDelimiterX{\norm}[1]{\lVert}{\rVert}{#1}
\DeclarePairedDelimiterX{\ceil}[1]{\lceil}{\rceil}{#1}

\newcommand{\N}{\mathbb{N}}
\newcommand{\Z}{\mathbb{Z}}
\newcommand{\R}{\mathbb{R}}
\newcommand{\Q}{\mathbb{Q}}
\newcommand{\B}{\mathbb{B}}

\newcommand{\FF}{\ensuremath{\mathds{F}}}

\newcommand{\AAA}{\ensuremath{\mathcal{A}}}
\newcommand{\BBB}{\ensuremath{\mathcal{B}}}
\newcommand{\CCC}{\ensuremath{\mathcal{C}}}
\newcommand{\FFF}{\ensuremath{\mathcal{F}}}
\newcommand{\III}{\ensuremath{\mathcal{I}}}
\newcommand{\NNN}{\ensuremath{\mathcal{N}}}
\newcommand{\PPP}{\ensuremath{\mathcal{P}}}
\newcommand{\SSS}{\ensuremath{\mathcal{S}}}

\newcommand{\cocl}[1]{\ensuremath{\operatorname{#1}}}
\newcommand{\W}[1]{\cocl{W[#1]}}
\newcommand{\Wone}{\W{1}}

\newcommand{\NP}{\cocl{NP}}

\newcommand{\bigO}{\mathcal{O}}
\newcommand{\yes}{\textnormal{\texttt{yes}}}
\newcommand{\no}{\textnormal{\texttt{no}}}

\DeclareMathOperator{\dist}{dist}

\DeclareMathOperator{\mmin}{in}

\newcommand{\repr}{\ensuremath{\subseteq_{\mathrm{rep}}}}
\newcommand{\Nin}{\ensuremath{N^{\mmin}}}

\newcommand{\low}{\ensuremath{\alpha}}
\newcommand{\upp}{\ensuremath{\beta}}
\newcommand{\digits}{\ensuremath{\tau}}

\newcommand{\oneto}[1]{[ #1 ]} %

\newcommand{\eps}{\varepsilon}

\newcommand{\abX}[3]{\ensuremath{#1}-\ensuremath{#2}\nobreakdash-#3}
\newcommand{\abpath}[2]{\abX{#1}{#2}{path}}

\newcommand{\stpath}{\abpath{s}{t}}

\newcommand{\bsp}{\textnormal{\textsc{Balance-Fair Shortest Path}}}
\newcommand{\splu}{\textnormal{\textsc{Short Path with Lower and Upper Bounds}}}
\newcommand{\mcc}{\textnormal{\textsc{Multicolored Clique}}}

\newcommand*{\defeq}{\mathrel{\vcenter{\baselineskip0.5ex \lineskiplimit0pt\hbox{\scriptsize.}\hbox{\scriptsize.}}}=}

\title{
	\Large\bf
	Fair Short Paths in Vertex-Colored Graphs%
	\footnote{%
		This work was initiated at the 2021 Research Retreat of the Algorithmics and Computational Complexity group, Technische Universität Berlin.
	}
}

\author{Matthias Bentert\textsuperscript{1}}
\author{Leon Kellerhals\textsuperscript{2}}
\author{%
	Rolf Niedermeier\textsuperscript{2,}%
	\thanks{
		We dedicate this paper to Rolf, who tragically passed away last year.
		We are deeply affected by this loss of our co-author, colleague, and advisor.
		Rolf contributed tremendously to computer science and, in particular, to parameterized algorithmics, and should have continued doing so for a long time.
		The computer science community shall build on the foundations he has laid.
	}
}
\affil{\small
\textsuperscript{1} University of Bergen, Department of Informatics, Algorithms\protect\\
  \texttt{matthias.bentert@uib.no}}

\affil{\small
\textsuperscript{2}  Technische Universit\"at Berlin, Faculty~IV, Institute of Software Engineering and Theoretical Computer Science, Algorithmics and Computational Complexity\protect\\
  \texttt{leon.kellerhals@tu-berlin.de}}

\date{}

\newcommand{\fair}{balance-fair}

\newcommand{\colorvar}{\ensuremath{r}}
\newcommand{\colo}[1]{\ensuremath{q_{#1}}}
\newcommand{\jokercolor}{\ensuremath{p}}

\begin{document}

\maketitle

\begin{abstract}
The computation of short paths in graphs with arc lengths is a pillar
of graph algorithmics and network science. In a more diverse world,
however, not every short path is equally valuable.
For the setting where each vertex is assigned to a group (color),
we provide a framework to model multiple natural fairness aspects.
We seek to find short paths in which the number of occurrences of each color
is within some given lower and upper bounds.
Among other results, we prove the introduced problems
to be computationally intractable (NP-hard and parameterized hard with respect to the number of colors)
even in very restricted settings (such as each color should appear with exactly the same frequency),
while also presenting an encouraging algorithmic result (``fixed-parameter tractability'') related to the length of
the sought solution path for the general problem.
\end{abstract}

\section{Introduction}
Travel agency WhataWonderfulWorld offers adventure bus trips from 
New Orleans to New York,
making stops at exciting country sites for impressive day trips throughout the journey.
To address environmental demands, the agency wishes to minimize the overall travel distance,
while at the same time striving to maximize the variety and balance of impressions gathered at the day trips.
Clearly, such a sustainable travel from a starting point~$s$ to an 
endpoint~$t$ can be modeled as finding an \stpath{} in 
a graph with positive arc lengths.
To model maximum variety and balance,
the vertices---the places to visit---are colored according to 
agency-chosen categories, say blue vertices are interesting for bathing, 
green vertices for hiking, et cetera.
The quest is to find a short travel path in which the colors are \emph{fairly} represented.
Herein, what is to be considered \emph{fair} may depend on the setting.
In an idealistic (and hypothetical) setting, one may consider it fair to have the same amount of beaches, hiking spots, and places of any other type on the path.
We call such paths balance-fair; an example for such a path is depicted in \cref{fig:introexample}.
Possibly, one has asked the travelers beforehand about their preferences.
With this, one gains a better understanding on how many places of each type should be visited and one may be able to give lower and upper bounds for each type.
The latter (more realistic) fairness constraint indeed generalizes upon multiple fairness concepts introduced in the literature.
We discuss these later in this section.

\begin{figure}
\centering
\begin{tikzpicture}[scale=.8]
		\node[colornode=blue, label=$s$] at (-4.5,0) (s) {};
		\node[colornode=green, label=$t$] at(4.5,0) (t) {};
		
		\node[colornode=blue] at (-3.5,-1) (p11) {};
		\node[colornode=green] at (-2.5,-1) (p12) {};
		\node[colornode=blue] at (-1.5,-1) (p13) {};
		\node[colornode=green] at (-.5,-1) (p14) {};
		\node[colornode=blue] at (.5,-1) (p15) {};
		\node[colornode=green] at (1.5,-1) (p16) {};
		\node[colornode=green] at (2.5,-1) (p17) {};
		\node[colornode=blue] at (3.5,-1) (p18) {};
		
		\node[colornode=green] at (-3.5,1) (p21) {};
		\node[colornode=blue] at (-2.5,1) (p22) {};
		\node[colornode=green] at (-1.5,1) (p23) {};
		\node[colornode=green] at (-.5,1) (p24) {};
		\node[colornode=blue] at (.5,1) (p25) {};
		\node[colornode=green] at (1.5,1) (p26) {};
		\node[colornode=blue] at (2.5,1) (p27) {};
		\node[colornode=green] at (3.5,1) (p28) {};

		\draw[->] (s) to (p11);
		\draw[->] (s) to (p21);
		\draw[->] (p11) to (p12);
		\draw[->] (p21) to (p12);
		\draw[->] (p21) to (p22);
		\draw[->] (p12) to (p13);
		\draw[->] (p12) to (p23);
		\draw[->] (p22) to (p13);
		\draw[->] (p22) to (p23);
		\draw[->] (p13) to (p14);
		\draw[->] (p13) to (p24);
		\draw[->] (p23) to (p24);
		\draw[->] (p14) to (p15);
		\draw[->] (p14) to (p25);
		\draw[->] (p24) to (p25);
		\draw[->] (p15) to (p16);
		\draw[->] (p15) to (p26);
		\draw[->] (p25) to (p16);
		\draw[->] (p16) to (p27);
		\draw[->] (p26) to (p17);
		\draw[->] (p17) to (p18);
		\draw[->] (p27) to (p18);
		\draw[->] (p27) to (p28);
		\draw[->] (p18) to (t);
		\draw[->] (p28) to (t);
		
		\draw[path] (s.center) -- (p11.center) -- (p12.center) -- (p23.center) -- (p24.center) -- (p25.center) -- (p16.center) -- (p27.center) -- (p18.center) -- (t.center);
	\end{tikzpicture}
	\caption{A graph with colored vertices (blue and green), unit-length arcs, and two vertices~$s$ and~$t$ is depicted. The highlighted path is a shortest path between~$s$ and~$t$ and contains five blue and five green vertices. Thus, it is balance-fair.}
	\label{fig:introexample}	
\end{figure}
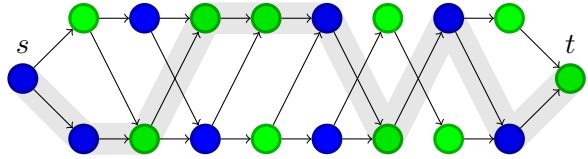%

\paragraph*{Our contributions.}
We introduce and study two natural fairness scenarios for
one of the back-bone problems in network algorithmics: finding
short paths.
The problems in consideration are the following.
\problemdef{\bsp{}}
{A directed graph~$G = (V, A)$, a vertex coloring~$\chi \colon V \to \oneto{c}$, an arc-length function~$w \colon A \rightarrow \N$, and two vertices~$s, t \in V$.}
{Is one of the shortest \stpath{s} \emph{\fair{}}, that is, it visits the same number of vertices of each color?}

\problemdef{\splu{}}
{A directed graph~$G = (V, A)$, a vertex coloring~$\chi \colon V \to \oneto{c}$, an arc-length function~$w \colon A \rightarrow \N$, two vertices~$s, t \in V$, and integers~$\ell, \low_1,\low_2,\ldots,\low_c,\upp_1,\upp_2,\ldots,\upp_c$.}
{Is there an~\stpath~$P$ of length at most~$\ell$ such that the number of vertices in~$P$ of color~$i$ is at least~$\low_i$ and at most~$\upp_i$ for each~$i \in \oneto{c}$?}

Simply put, \splu{} allows for very general fairness constraints, while \bsp{} focuses on a seemingly more simple (and possibly idealistic) constraint and only allows for shortest paths.
Note that \bsp{} is indeed a special case of \splu{} by setting~$\ell = \dist(s, t)$ and fixing all lower and upper bounds to~$(\ell-1)/c$.
Already the special case \bsp{} turns out to be NP-hard in general.
We cope with this computational intractability by investigating the
parameterized complexity of these two problems with respect to the two
perhaps most natural parameters: the number of different 
vertex colors in the input graph
and the length of the sought-after solution path.
Despite the fact that \bsp{} seems much more restricted than \splu{} at first glance, we only identify very minor differences in their (parameterized) complexity, that is, we show most of our hardness results for \bsp{} and most of our positive results for \splu.
Our results suggest that the difference in complexity is due to asking for a shortest path versus asking for any path, but not due to the more general fairness criterion.
We prove that finding a \fair{} shortest path is \Wone-hard for the parameter
number of colors and that such paths can be found in polynomial time for any constant number of colors.
If, however, we do not require the solution to be a shortest path, then the problem becomes NP-hard for only two colors.
Indeed, both hardness results hold even when all arcs have the same lengths.
Lastly, we show that the more general variant \splu{} is fixed-parameter tractable with respect to the length of the path.
The algorithm is based on representative families.

\paragraph*{Related work.}
Path finding in vertex-colored graphs has been a subject of broad and 
intensive study. Here, we only point to algorithmically motivated work that seems particularly close
to our scenario. %
First of all, a colorful path (sometimes also called a rainbow path or a multicolored path) is a path containing each color at most once and finding colorful paths is an important algorithmic 
topic, both in static and in temporal graphs~\cite{alon1995color,DH21}.
\citet{CIMTP21} analyze the complexity of finding 
tropical paths, that is, paths containing at least one vertex of each color.
They provide both tractability and intractability results.

Another close (and also broad) area is that of finding 
resource-constrained shortest paths, where, roughly speaking,
the desired path shall have minimum cost and only a limited consumption
of resources. Generally, the problem is NP-hard~\cite{HZ80} and
it has been extensively studied over the years \cite{ford2022backtracking,Irnich2005,PG13}.
Given its intractability, several heuristics have 
been proposed recently~\cite{ATHK21}. 
These models however do not have fairness aspects in mind.
\citet{hanaka2021computing} are somewhat closer to fairness aspects in path finding.
They study shortest paths under diversity aspects,
meaning that they search for multiple shortest paths that 
are maximally different from each other; this fits into the recent trend of
finding diverse sets of 
solutions to optimization problems~\cite{BFJMOPR20,PT19,kellerhals2021parameterized}.

Finally, we only mention in passing that fairness aspects
are currently investigated
in all kinds of optimization problems (particularly graph-based ones), 
including topics such as graph-based data 
clustering~\cite{AEKMMPVW20,AEKM20,FM21,FKN21}, 
influence maximization~\cite{KRBHJ20}, matching~\cite{CKLV19}, and graph mining~\cite{DMCL22,KT21}.

\paragraph*{Fairness measures.}
Our fairness concept---ensuring that the number of occurrences of each color in the solution is within a given range---is closely related to that introduced by \citet{celis2018ranking} to find fair rankings.
When considering paths, our constraints generalize finding rainbow and tropical graphs.
It also generalizes upon many recently introduced fairness variants.
We list three examples.
\begin{asparadesc}
	\item[\rm\itshape Max-min fairness.]
		Defined as the difference or quotient between the number of occurrences of the most frequent and the least frequent color in the solution, this concept generalizes our \fair{ness} (for which the difference is~$0$ and the quotient is~$1$).
		This variant has been used in many works, including the seminal work on fair clustering by \citet{chierichetti2017clustering}.
		We can find solutions with a max-min fairness of~$k$ by guessing\footnote{Whenever we pretend to guess something, we iterate over all possibilities and consider the correct iteration.} the number~$\alpha$ of occurrences of the least frequent colors.
		All lower bounds are then set to~$\alpha$ and all upper bounds are set to~$\alpha + k$ or~$\alpha k$.
	\item[\rm\itshape Proportional fairness.]
		The goal of this variant is to make each color appear roughly with the same frequency in the solution as in the input.
		This is a standard axiom in the area of fair division and was used in graph-based clustering~\cite{FKN21} and principal
		component analysis~\cite{samadi2018pca} to ensure a balanced error or distortion among the colors.
		We can model this fairness variant in the same way as max-min fairness.
	\item[\rm\itshape Margin-of-victory fairness.]
		Herein, one minimizes the difference in occurrences between the first and second most frequent color in the solution.
		This more relaxed fairness notion prevents a color from becoming a dominating majority in the solution.
		It was incorporated into the problem of finding fair many-to-one matchings~\cite{stoica2020mov,boehmer2022fairmatching}.
		To model this variant, we guess the two most frequent colors as well as their number of occurrences.
		The bounds for each color can then be trivially obtained.
\end{asparadesc}

\section{Preliminaries}
\label{sec:prelim}
We denote by~$\N$ the set of all positive integers and by~$\N_0$ the set of all non-negative integers.
For an integer~$n \in \N$, let $\oneto{n} \defeq \{1, 2, \dots, n\}$.

\paragraph{Graphs.}
We use standard graph-theoretic terminology.
All graphs are directed if not explicitly stated otherwise.
For a directed graph~${G = (V, A)}$, we set~$n \defeq \abs{V}$ and~$m \defeq \abs{A}$.
For a vertex~$v \in V$, we denote by~$\Nin(v)$ the set of all vertices~$u$ such that~$(u,v) \in A$.
A \emph{path}~$P$ \emph{on~$\ell$~vertices} is a graph with vertex 
set~$\{v_1, v_2, \ldots, v_{\ell}\}$
which contains the arc~$(v_i, v_{i+1})$ for each~$i \in \oneto{\ell-1}$.
The vertex~$v_1$ is called the \emph{startpoint} and~$v_{\ell}$ is called the \emph{endpoint} of~$P$.
We also say that~$P$ is a path from~$v_1$ to~$v_\ell$.
We denote by~$V(P)$ the set~$\{v_1, v_2, \ldots, v_{\ell}\}$ of vertices in~$P$.
Let~$G = (V, A)$ be a graph with two vertices~$s$ and~$t$ and~$w \colon A \to \N$ be an arc-length function.
An~\stpath{}~$P$ is a subgraph of $G$ which is a path from~$s$ to~$t$.
The \emph{length}~$w(P)$ of the path is the sum of its arc lengths.
We denote by~$\dist_G(s, t)$ the length of a shortest~\mbox{\stpath{}} in $G$.
Whenever clear from context, we may drop the subscript~$G$.
If not stated otherwise, we assume arc lengths to be positive.

For a graph $G = (V, A)$, a vertex coloring~$\chi \colon V \to \oneto{c}$, and a color~$i \in \oneto{c}$, we denote by~${\chi^i \defeq \{v \in V \mid \chi(v)=i\}}$ the set of vertices of color~$i$.
For a subgraph~$H$ of~$G$ and a color~$i$, we denote by~$\chi_H$ the coloring~$\chi$ restricted to the vertices of~$H$, and by~$\chi_H^i$ the set of vertices of color~$i$ in~$H$.

\paragraph{Matroids.}
A pair~$M = (U, \III)$, where~$U$ is called \emph{ground set} and $\III$ is a family of subsets (called \emph{independent sets}) of~$U$,
is a \emph{matroid} if
\begin{inparaenum}[(i)]
	\item $\emptyset \in \III$,
	\item if~${A' \subseteq A}$ and~${A \in \III}$, then~${A' \in \III}$ (hereditary property), and
	\item if~${A, B \in \III}$ and~${\abs{A} < \abs{B}}$, then there is an~${e \in B \setminus A}$ such that~${A \cup \{e\} \in \mathcal I}$ (exchange property).
\end{inparaenum}
An inclusion-wise maximal independent set is a \emph{basis} of~$M$.
It follows from the exchange property that all bases of~$M$ have the same size.
This size is called the \emph{rank} of~$M$.
Let~$A$ be a matrix over a finite field~$\FF$, and let~$U$ be the set of columns of $A$.
We associate a matroid~$M = (U, \III)$ with $A$ as follows.
A set~$X \subseteq U$ is independent (i.e., $X \in \III$) if the columns in $X$ are linearly independent over $\FF$.
We say that the matroid is \emph{linear} and that~$A$ represents $M$.

\emph{Gammoids} are a family of matroids defined as follows.
Given a directed graph~$G=(V, A)$ with vertex subsets~$S, T \subseteq V$,
we say that~$X \subseteq T$ is \emph{linked} to~$S$ if there exist $\abs{X}$~vertex-disjoint paths (possibly of length~$0$) going from~$S$ to~$X$.
The \emph{gammoid}~$M = (T, \III)$ corresponding to~$G$, $S$ has ${\III = \{X \subseteq T \mid X \text{ is linked to } S \}}$ as its family of independent sets.
\begin{theorem}[\citet{kratsch2014compression}]
	\label{thm:gammoid-repr}
	Let~$\eps > 0$, let~$G=(V,A)$ be a directed graph, let~$S, T \subseteq V$, and let~${M = (T, \III)}$ be the gammoid corresponding to~$G$, $S$, and~$T$.
	Then, one can compute in polynomial time an $\abs{S}\times \abs{T}$ matrix~$A$ over the rationals which represents a matroid~${\widetilde M = (T, \widetilde\III)}$ such that for any~$X \subseteq T$,
	we have~$X \notin \widetilde\III$ whenever~$X \notin \III$ and~$\Pr[X \in \widetilde\III] \ge 1-\eps$ whenever~$X \in \III$.
	Moreover, the entries in~$A$ are of bit-length $\bigO(\min\{\abs{T}, \abs{S} \log \abs{T}\} + \log(1/\eps) + \log \abs{V})$.
\end{theorem}

\paragraph{Parameterized complexity.}
A parameterized problem is \emph{fixed-parameter tractable} if there exists an algorithm solving any instance~$(I, \rho)$ ($I$ is the input instance and $\rho$ is some parameter) in $f(\rho) \cdot \abs{I}^{\bigO(1)}$ time, where $f$ is a (computable) function solely depending on~$\rho$.
To show that a parameterized problem $L'$ is presumably not fixed-parameter tractable, one may use a \emph{parameterized reduction} from a \Wone-hard problem to~$L$.
A parameterized reduction from a parameterized problem $L$ to another parameterized problem~$L'$ is a function satisfying the following.
There are two computable functions~$f$ and~$g$, such that given an instance~$(I, \rho)$ of~$L$, the reduction computes in~$f(\rho) \cdot \abs{I}^{\bigO(1)}$~time an instance~$(I', \rho')$ of $L'$ such that~$\rho' \leq g(\rho)$ and~$(I, \rho)$ is a \yes-instance of~$L$ if and only if~$(I', \rho')$~is a \yes-instance of~$L'$.

\section{The Parameter Number of Colors}
\label{sec:colors}
In this section, we study the computational complexity of our two problems parameterized by the number of colors.
Recall that we assume that all arc lengths are positive.
We justify this assumption by showing that, if arc lengths may be zero, \bsp{} with two colors becomes NP-hard.

\begin{restatable}{observation}{NPZero}
	\label{prop:0-arc-np}
	\bsp{} is \NP-hard for two colors when zero-length arcs are allowed.
\end{restatable}

\begin{proof}
We provide a reduction from the well-known problem \textsc{Directed Hamiltonian $s$-$t$-Path}.
Given a directed graph and two vertices~$s$ and~$t$, the question is whether the graph contains an \stpath{} visiting all vertices.
The problem is known to be NP-hard \cite{GareyJ79}.
Given an instance~${I = (G=(V,A), s, t)}$, we construct an equivalent instance~$I'$ of \bsp{} with~${I' = (G'=(V',A'),\chi,w,s',t')}$ as follows.
The graph~$G'$ consists of~$G$ plus a $\abs{V}$-vertex path~$Q$ attached to~$t$.
The endpoint of~$Q$ is~$t'$.
We identify~$s'$ with~$s$.
All vertices except for those in~$Q$ are colored with the first color.
All arcs in~$G'$ have length zero.
Now any \abpath{s'}{t'} contains all~$\abs{V}$ vertices in~$Q$, which have the second color.
Hence, it is \fair{} if and only if it contains an~$s$-$t$-subpath that visits all vertices in~$V$.
\end{proof}

Next, we look at the more general \splu.
Note that the problem coincides with \prob{Directed Hamiltonian $s$-$t$-Path} when there is one color, the arcs have unit length, and the lower bound is equal to the number of vertices.
This yields the following.

\begin{restatable}{observation}{NPl}
	\label{prop:flp-l-np}
	\splu{} is \NP-hard even with unit arc lengths and only one vertex color.
\end{restatable}

The reduction in \cref{prop:flp-l-np} makes use of the fact that one may take detours to satisfy the lower bounds.
Hence, it does not work for the special case \bsp{} or if we require that~$\ell = \dist(s,t)$.
When enforcing the solution to be a shortest \stpath{}, we can show the following.

\begin{theorem}
	\label{prop:cXP}
	\splu{} is solvable in $\bigO(n^{c}\cdot m)$ time, when~$\ell = \dist(s,t)$ and where~$c$ is the number of colors.
\end{theorem}
\begin{proof}
	We devise a dynamic program with a Boolean table~${T \colon V \times [n]^{c} \rightarrow \{0, 1\}}$ storing for each vertex~$v$ and each tuple $(x_1,x_2,\ldots,x_{c})$ whether there is a shortest~\abpath{s}{v} in~$G$ which contains exactly~$x_i$ vertices of color~$i$.
	We say that such paths \emph{respect} the tuple~$(x_1,x_2,\ldots,x_{c})$.
	The table is computed for all vertices in order of their distances from~$s$.
	Note that it holds for each arc~$(u, v)$ in any shortest~\stpath~$P$, that the vertex~$u$ is closer to~$s$ than~$v$ and~${\dist(s, v) = \dist(s, u) + w((u, v))}$.
	Any arc not satisfying this equality cannot be part of a shortest~\stpath{} and can therefore be deleted.
	Denote by
	\begin{align*}
	A'\! \defeq \{(u, v) \in A \mid \dist_G(s, v) = \dist_G(s, u) + w((u, v))\}
	\end{align*}
	the set of remaining arcs.
	To compute an entry~$T[v,x_1,x_2,\ldots,x_{c}]$, we fo the following.
	Let~$i$ denote the color of~$v$.
	We iterate over all incoming arcs~$(u,v) \in A'$ of~$v$ and compute
	\[
		T[v,x_1,x_2,\ldots,x_{c}] = \bigvee_{(u,v) \in A'} T[u,x'_1,x'_2,\ldots,x'_{c}],
	\]
	wherein~$x'_i = x_i - 1$ and $x'_j = x_j$ for all~$j \ne i$.
	That is, $T[v,x_1,x_2,\ldots,x_{c}]$ is~set to true if and only if there is a vertex~$u \in \Nin(v)$ for which $T[u,x'_1,x'_2,\ldots,x'_{c}]$ is true.
	If~$T[v,x_1,\ldots,x_{c}]$ is set to true, then there is a shortest \mbox{\abpath{s}{u}} which respects the tuple~$(x'_1,x'_2,\ldots,x'_{c})$.
	Appending~$v$ to this results in a path respecting~$(x_1,x_2,\ldots,x_{c})$ since~$x_i = x'_i+1$.
	By the definition of~$A'$, this path is also a shortest \abpath{s}{v}.
	In the other direction, assume that there is a shortest~\abpath{s}{v}~$P$ respecting~$(x_1,x_2,\ldots,x_{c})$.
	Consider the penultimate vertex~$u$ in~$P$ and the subpath~$P'$ from~$s$ to~$u$.
	Note that since~$P$ and~$P'$ only differ in~$v$, it holds that~$P'$ is a shortest~\abpath{s}{u} respecting~$(x'_1,x'_2,\ldots,x'_{c})$, where~$x'_i = x_i-1$ and~$x'_j = x_j$ for all colors~$j \ne i$.
	Thus,~$T[u,x'_1,x'_2,\ldots,x'_{c}]$ is set to true and thus by construction also~$T[v,x_1,x_2,\ldots,x_{c}]$.
	
	Once the whole table is computed, we check whether there is a tuple~$(x_1,x_2,\ldots,x_{c})$ such that~$T[t,x_1,x_2,\ldots,x_{c}]$ is set to true and~$\alpha_i \leq x_i \leq \beta_i$ for all~$i$.
	Lastly, we analyze the running time.
	Observe that there are~$n^{c}$ table entries for each vertex~$v$ and computing one table entry takes~$\bigO(|\Nin(v)|)$ time.
	Thus, the overall running time is in~$\bigO(m \cdot n^{c})$.
\end{proof}

We next show that there is little hope for a significantly better algorithm for \bsp{} parameterized by the number~$c$ of colors.
We prove that \bsp{} is \Wone-hard with respect to~$c$ and cannot be solved in~$n^{o(\nicefrac{c}{\log c})}$~time unless the Exponential Time Hypothesis (ETH) fails.
The ETH states that \textsc{3-SAT} cannot be solved in subexponential time \cite{IP99}.
Both results are shown via reductions from \prob{Multicolored Clique}:
Given a $k$-partite undirected graph~$G = (V, E)$ with partitions~$V_1, V_2, \dots, V_k$, does~$G$ contain a clique on~$k$ vertices?
\prob{Multicolored Clique} is known to be \Wone-complete~\cite{pietrzak2003mcc}.

\begin{restatable}{theorem}{WOne}
	\label{thm:cW1}
	\bsp{} is \Wone-hard{} when parameterized by the number~$c$ of colors.
\end{restatable}
\begin{proof}
	We derive a parameterized reduction from \mcc{} parameterized by solution size~$k$.
	To this end, let~$G=(V,E)$ be the input graph of our instance of \mcc{} with partitions~$V_1, V_2, \dots, V_k$.
	For the sake of simplicity, we will assume that each~$V_i$ has the same number~$\eta \defeq \abs{V}/k$ of vertices (this is no restriction since we can simply add isolated vertices) and that~$V_i = \{v_1^i,v_2^i,\ldots, v_\eta^i\}$.
	
	In the following, we will construct a graph~$H$ with two vertices~$s$ and~$t$ that contains a \fair{} shortest \stpath{} if and only if~$G$ contains a clique of size~$k$.
	Let us first give an intuitive description of the different pieces.
	The graph~$H$ will be made mostly from two parts: a vertex-selection gadget for each partition~$V_i$ and an edge-verification gadget for each pair~$V_i \ne V_{j}$ of partitions.
	The former (broadly speaking) decides for each $i \in \oneto{k}$ which vertex of~$V_i$ is supposed to be in the clique.
	The latter verifies that there is an edge between the two chosen vertices of the two respective partitions.
	We need~$2k(k-1)+1$ colors: two colors~$\colorvar_{i,j}$ and~$\colo{i,j}$ for each~$i \in [k]$ and each~$j \in [k] \setminus \{i\}$ and a special color~$\jokercolor$.
	
	We start by introducing~$K \defeq k+\binom{k}{2}+1$ vertices~$u_1,u_2,\ldots,u_{K}$ of color~$\jokercolor$.
	The vertex-selection gadget for each partition~$V_i$ consists of~$\eta$ vertex-disjoint paths from~$u_{i}$ to~$u_{i+1}$.
	Each of these paths represents one vertex~$v_a^i$ and contains~$a$ vertices of color~$\colorvar_{i,j}$ and~$\eta - a$ vertices of color~$\colo{i,j}$ for each~$j \in [k] \setminus \{i\}$.
	The edge-verification gadget for each pair of partitions~$V_i \neq V_j$ also consists of vertex-disjoint paths which are from~$u_{d}$ to~$u_{d+1}$ for some~$d$ such that each gadget uses a different ``slot''~$d$.
	Each of these vertex-disjoint paths represents one edge~$\{v_a^i,v_b^j\}$ and contains~$\eta-a$ vertices of color~$\colorvar_{i,j}$,~$a$ vertices of color~$\colo{i,j}$,~$\eta-b$~vertices of color~$\colorvar_{j,i}$, and~$b$ vertices of color~$\colo{j,i}$.
	To complete the construction, we add~$2\eta-K$ vertices of color~$\jokercolor$ if~$2\eta \geq K$ and~$K-2\eta$ vertices of each color~$\colorvar_{i,j}$ and~$\colo{i,j}$ for each~$i \neq j \in [k]$ if~$2\eta < K$.
	All of these new vertices form a directed path.
	We add an arc from~$u_K$ to the start of this path and call the last vertex~$t$.
	Finally, we rename~$u_1$ to~$s$ and set all arc lengths to one.
	
	Since the reduction takes polynomial time, it remains to show the correctness.
	To this end, first assume that there is a multicolored clique~$C = \{v_{a_1}^1,v_{a_2}^2,\ldots,v_{a_k}^k\}$ in~$G$.
	Then, there is a \fair{} shortest~$s$-$t$-path~$P$ in~$H$ that uses the path representing~$v_{a_i}^i$ in the vertex-selection gadget for~$V_i$, the path representing the edge~$\{v_{a_i}^i,v_{a_j}^j\}$ in the edge-verification gadget for~$V_i$ and~$V_j$, and ends with the path from~$u_K$ to~$t$.
	Note that since~$C$ is a clique, all of these paths exist in the edge-verification gadgets and that each color except for~$\jokercolor$ appears exactly~$2\eta$ times in all vertex-selection and edge-verification gadgets combined.
	The color~$\jokercolor$ appears~$K$ times up to vertex~$u_K$.
	By the construction of the final path, the whole path~$P$ is then \fair.
	
	Now assume that there is a \fair{} \stpath~$P$ in~$H$.
	Let~$v_{b_i}^i$ be the vertex represented by the subpath of~$P$ through the vertex-selection gadget for~$V_i$ for each~${i \in [k]}$.
	Note that~$P$ contains~$b_i$ vertices of color~$\colorvar_{i,j}$ in this subpath for each~${j\in[k]\setminus \{i\}}$.
	In order to be \fair, it needs to contain exactly~$\eta-b_i$ vertices of color~$\colorvar_{i,j}$ in the edge-verification gadget for~$V_i$ and~$V_j$ for each~$j\in [k]\setminus \{i\}$.
	Thus, by the construction of the edge-verification gadgets, the vertices~$v_{b_i}^i$ for all~$i\in [k]$ need to be connected and thus form a multicolored clique.
\end{proof}

We continue with the lower bound based on the Exponential Time Hypothesis (ETH).
If the hypothesis is true, then \textsc{Multicolored Clique} parameterized by solution size~$k$ cannot be solved in~$f(k) \cdot n^{o(k)}$ time for any computable function~$f$ \cite{LMS18}.

\begin{theorem}
	\label{thm:eth}
	\bsp{} parameterized by the number~$c$ of colors cannot be solved in~${f(c) \cdot n^{o(\nicefrac{c}{\log c})}}$~time for any computable function~$f$ unless the ETH breaks.
\end{theorem}

\begin{proof}
	We reduce from \mcc{} parameterized by solution size~$k$.
	To this end, let~$G=(V,E)$ be the input graph of our instance of \mcc{} with partitions~$V_1, V_2, \dots, V_k$.
	For the sake of simplicity, we will again assume that each~$V_i$ has the same number~$\eta \defeq n/k$ of vertices.
	Let~$V_i = \{v_0^i,v_1^i,\ldots, v_{\eta-1}^i\}$.
	\cref{fig:mcc} gives an example used throughout this proof.
	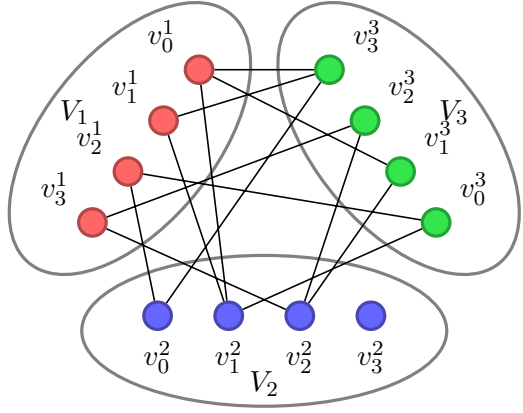
\begin{figure}[t]
	\centering
	\begin{tikzpicture}[yscale=1,xscale=1.2]

	\path
		(90:3cm) to
			node[pos=0.275, colornode=r, label=150:$v_0^1$] (u1) {}
			node[pos=0.425, colornode=r, label=150:$v_1^1$] (u2) {}
			node[pos=0.5] (eu) {}
			node[pos=0.575, colornode=r, label=150:$v_2^1$] (u3) {}
			node[pos=0.725, colornode=r, label=150:$v_3^1$] (u4) {}
		(210:3cm) to
			node[pos=0.275, colornode=b, label=270:$v_0^2$] (v1) {}
			node[pos=0.425, colornode=b, label=270:$v_1^2$] (v2) {}
			node[pos=0.575, colornode=b, label=270:$v_2^2$] (v3) {}
			node[pos=0.725, colornode=b, label=270:$v_3^2$] (v4) {}
		(330:3cm) to
			node[pos=0.275, colornode=g, label=030:$v_0^3$] (w1) {}
			node[pos=0.425, colornode=g, label=030:$v_1^3$] (w2) {}
			node[pos=0.575, colornode=g, label=030:$v_2^3$] (w3) {}
			node[pos=0.725, colornode=g, label=030:$v_3^3$] (w4) {}
		(90:3cm);
	\draw[very thick, gray] (150:1.7cm) ellipse [x radius=2cm,y radius=1cm,rotate=060]
		(270:1.7cm) ellipse [x radius=2cm,y radius=1cm,rotate=180]
		(030:1.7cm) ellipse [x radius=2cm,y radius=1cm,rotate=300];

	\node at (150:2.4cm) {$V_1$};
	\node at (270:2.4cm) {$V_2$};
	\node at (030:2.4cm) {$V_3$};
	
	\draw[semithick] (u1) -- (v2)
		(u2) -- (v2)
		(u3) -- (v1)
		(u1) -- (w2)
		(u1) -- (w4)
		(u3) -- (w1)
		(u4) -- (w3)
		(v1) -- (w4)
		(v2) -- (w1)
		(v3) -- (w2)
		(v3) -- (w3)
		(u2) -- (w4)
		(u4) -- (v3);
\end{tikzpicture}
	\caption{An example instance of \mcc{} with~$k = 3$ colors and~$\eta=4$ vertices of each color.}
	\label{fig:mcc}
	\end{figure}
	
	In the following, we construct a graph~$H$ with two vertices~$s$ and~$t$ that contains a \fair{} shortest \stpath{} if and only if~$G$ contains a clique of size~$k$.
	Again, $H$ contains a vertex-selection gadget for each partition~$V_i$ and an edge-verification gadget for each pair~$V_i \ne V_{j}$ of partitions.
	Our reduction will use the colors~$\colorvar_i$ and~$\colo{i}$ for each~$i \in \oneto{k}$ and a special color~$\jokercolor$.
	The main idea behind the reduction is to represent each vertex~$v_a^i$ with~$k^{\bigO(\log n)}$ vertices of color~$\colorvar_{i}$ and/or~$\colo{i}$ using a binary encoding of~$a$.
	In the vertex-selection gadgets, we will represent vertices of~$G$ by paths of vertices such that if this path is part of the final \fair{} path $P$, then $P$ contains exactly~$k-1$ paths in edge-verification gadgets that correspond to edges with~$v_a$ as an endpoint.
	
	As in the proof of \cref{thm:cW1}, we place~$K \defeq k + \binom{k}{2}+1$ vertices~$u_1,u_2,\ldots,u_K$ of color~$\jokercolor$ on a line and place the different gadgets in between.
	For each partition~$V_i$, we create a vertex-selection gadget as follows (see \cref{fig:vertexgadget} for an illustration).
	The gadget consists of~$\eta$ vertex-disjoint paths from~$u_i$ to~$u_{i+1}$.
	Each of these paths represents one vertex~$v_j^i$ and is constructed using the binary encoding of~$j$.
	To this end, let~$\digits \defeq \lceil \log_2 \eta \rceil$ be the number of digits required to encode~$\eta$.
	For each~$\ell \in \oneto{\digits}$ let~$d_\ell$ be the~$\ell$-th least significant bit in the binary encoding of~$j$.
	Then, we add~$(k-1)\cdot k^{\ell-1}$ vertices
	and color them with~$\colorvar_i$ if~$d_\ell = 0$ and with~$\colo{i}$ if~$d_\ell = 1$.
	We call the set of vertices added for a specific~$\ell$ the \emph{$\ell$-th level} of this gadget.
	Denote by
	$ %
	x \defeq (k-1) \cdot \sum_{\ell=1}^{\digits} k^{\ell-1} = k^\tau - 1
	$ %
	the number of vertices in any of the above vertex-disjoint paths.
	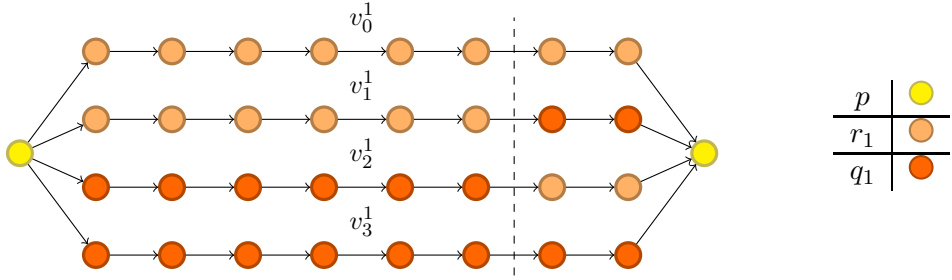
\begin{figure}[t]
	\centering
	\begin{minipage}{0.70\columnwidth}
	\begin{tikzpicture}[xscale=1.0, yscale=0.9,every node/.style={scale=0.9}]
		\node[colornode=yellow] at (-8.5,0) (s) {};
		\node[colornode=yellow] at(.5,0) (t) {};
		
		\node at(-4,2) {$v_0^1$};
		\node[colornode=r0] at (-7.5,1.5) (q11) {};
		\node[colornode=r0] at (-6.5,1.5) (q12) {};
		\node[colornode=r0] at (-5.5,1.5) (q13) {};
		\node[colornode=r0] at (-4.5,1.5) (q14) {};
		\node[colornode=r0] at (-3.5,1.5) (q15) {};
		\node[colornode=r0] at (-2.5,1.5) (q16) {};
		\node[colornode=r0] at (-1.5,1.5) (q17) {};
		\node[colornode=r0] at (-.5,1.5) (q18) {};
		
		\node at(-4,1) {$v_1^1$};
		\node[colornode=r0] at (-7.5,.5) (q21) {};
		\node[colornode=r0] at (-6.5,.5) (q22) {};
		\node[colornode=r0] at (-5.5,.5) (q23) {};
		\node[colornode=r0] at (-4.5,.5) (q24) {};
		\node[colornode=r0] at (-3.5,.5) (q25) {};
		\node[colornode=r0] at (-2.5,.5) (q26) {};
		\node[colornode=r1] at (-1.5,.5) (q27) {};
		\node[colornode=r1] at (-.5,.5) (q28) {};
		
		\node at(-4,0) {$v_2^1$};
		\node[colornode=r1] at (-7.5,-.5) (q31) {};
		\node[colornode=r1] at (-6.5,-.5) (q32) {};
		\node[colornode=r1] at (-5.5,-.5) (q33) {};
		\node[colornode=r1] at (-4.5,-.5) (q34) {};
		\node[colornode=r1] at (-3.5,-.5) (q35) {};
		\node[colornode=r1] at (-2.5,-.5) (q36) {};
		\node[colornode=r0] at (-1.5,-.5) (q37) {};
		\node[colornode=r0] at (-.5,-.5) (q38) {};
		
		\node at(-4,-1) {$v_3^1$};
		\node[colornode=r1] at (-7.5,-1.5) (q41) {};
		\node[colornode=r1] at (-6.5,-1.5) (q42) {};
		\node[colornode=r1] at (-5.5,-1.5) (q43) {};
		\node[colornode=r1] at (-4.5,-1.5) (q44) {};
		\node[colornode=r1] at (-3.5,-1.5) (q45) {};
		\node[colornode=r1] at (-2.5,-1.5) (q46) {};
		\node[colornode=r1] at (-1.5,-1.5) (q47) {};
		\node[colornode=r1] at (-.5,-1.5) (q48) {};
		
		\draw[->] (s) to (q11);
		\draw[->] (s) to (q21);
		\draw[->] (s) to (q31);
		\draw[->] (s) to (q41);
		\draw[->] (q11) to (q12);
		\draw[->] (q21) to (q22);
		\draw[->] (q31) to (q32);
		\draw[->] (q41) to (q42);
		\draw[->] (q12) to (q13);
		\draw[->] (q22) to (q23);
		\draw[->] (q32) to (q33);
		\draw[->] (q42) to (q43);
		\draw[->] (q13) to (q14);
		\draw[->] (q23) to (q24);
		\draw[->] (q33) to (q34);
		\draw[->] (q43) to (q44);
		\draw[->] (q14) to (q15);
		\draw[->] (q24) to (q25);
		\draw[->] (q34) to (q35);
		\draw[->] (q44) to (q45);
		\draw[->] (q15) to (q16);
		\draw[->] (q25) to (q26);
		\draw[->] (q35) to (q36);
		\draw[->] (q45) to (q46);
		\draw[->] (q16) to (q17);
		\draw[->] (q26) to (q27);
		\draw[->] (q36) to (q37);
		\draw[->] (q46) to (q47);
		\draw[->] (q17) to (q18);
		\draw[->] (q27) to (q28);
		\draw[->] (q37) to (q38);
		\draw[->] (q47) to (q48);
		\draw[->] (q18) to (t);
		\draw[->] (q28) to (t);
		\draw[->] (q38) to (t);
		\draw[->] (q48) to (t);
		
		\draw[dashed] (-2,-1.8) to (-2,2.0);

	\end{tikzpicture}%
	\end{minipage}
	~
	\begin{minipage}{0.25\columnwidth}
	\begin{tabular}{c|c}
		$\jokercolor$ & 
		\tikz{\node[colornode=yellow,thick,inner sep=3pt] at (0,0) (s) {};}\\\hline
		$\colorvar_1$ & \tikz{\node[colornode=r0,thick,inner sep=3pt] at (0,0) (s) {};}\\\hline
		$\colo{1}$ & \tikz{\node[colornode=r1,thick,inner sep=3pt] at (0,0) (s) {};}
	\end{tabular}
	\end{minipage}
	\caption{
	The vertex-selection gadget for~$V_1$ in \cref{fig:mcc} and a legend providing the names of the colors.
	Each path represents a vertex in~$V_1$ and the respective vertex is listed above each path.
	Vertices right of the dashed line belong to the first layer while vertices to the left belong to the second layer.}
	\label{fig:vertexgadget}
	\end{figure}

	The edge-verification gadget for a pair~$V_i, V_{j}$ of partitions is again a collection of vertex-disjoint parallel paths from~$u_d$ to~$u_{d+1}$ for some~$d$ (see \cref{fig:edgegadget} for an illustration).
	There is one path for each edge~$\{v_a^i, v_b^j\}$ in~$G$.
	For each~$\ell \in \oneto{\digits}$, let~$d_\ell$ be the~$\ell$-th least significant bit in the binary encoding of~$a$ and let~$d'_\ell$ be the~$\ell$-th least significant bit in the binary encoding of~$b$.
	We then add~$k^{\ell-1}$~vertices and color them with~$\colo{i}$ if~$d_\ell = 0$ and with color~$\colorvar_{i}$ if~$d_\ell = 1$.
	Analogously, we add~$k^{\ell-1}$ vertices and color them with~$\colo{j}$ if~$d'_\ell = 0$ and with ~$\colorvar_{j}$ if~$d'_\ell = 1$.
	We again call the set of vertices added for a specific~$\ell$ the $\ell$-th level of this gadget
	and denote by
	$ %
		y \defeq 2 \cdot \sum_{\ell=1}^{\digits} k^{\ell-1}
	$ %
	the number of vertices in any of the above vertex-disjoint paths.
	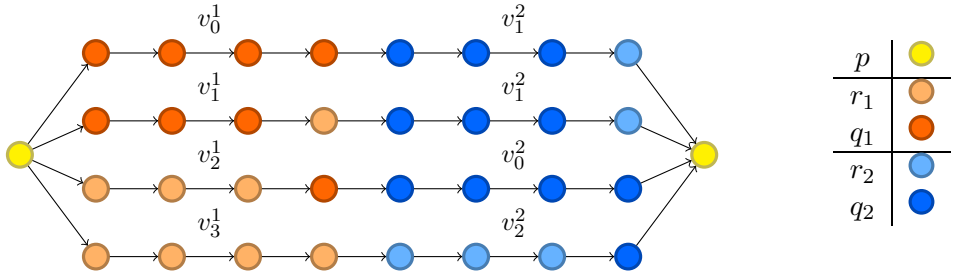
\begin{figure}[t]
	\begin{center}
	\begin{minipage}{0.70\columnwidth}
	\begin{tikzpicture}[yscale=.9,xscale=1,every node/.style={scale=0.9}]
		\node[colornode=yellow] at (-4.5,.5) (s) {};
		\node[colornode=yellow] at(4.5,.5) (t) {};
		
		\node at(-2,2.5) {$v_0^1$};
		\node at(2,2.5) {$v_1^2$};
		\node[colornode=r1] at (-3.5,2) (p11) {};
		\node[colornode=r1] at (-2.5,2) (p12) {};
		\node[colornode=r1] at (-1.5,2) (p13) {};
		\node[colornode=r1] at (-.5,2) (p14) {};
		\node[colornode=b1] at (.5,2) (p15) {};
		\node[colornode=b1] at (1.5,2) (p16) {};
		\node[colornode=b1] at (2.5,2) (p17) {};
		\node[colornode=b0] at (3.5,2) (p18) {};
		
		\node at(-2,1.5) {$v_1^1$};
		\node at(2,1.5) {$v_1^2$};
		\node[colornode=r1] at (-3.5,1) (p21) {};
		\node[colornode=r1] at (-2.5,1) (p22) {};
		\node[colornode=r1] at (-1.5,1) (p23) {};
		\node[colornode=r0] at (-.5,1) (p24) {};
		\node[colornode=b1] at (.5,1) (p25) {};
		\node[colornode=b1] at (1.5,1) (p26) {};
		\node[colornode=b1] at (2.5,1) (p27) {};
		\node[colornode=b0] at (3.5,1) (p28) {};
		
		\node at(-2,.5) {$v_2^1$};
		\node at(2,.5) {$v_0^2$};
		\node[colornode=r0] at (-3.5,0) (p31) {};
		\node[colornode=r0] at (-2.5,0) (p32) {};
		\node[colornode=r0] at (-1.5,0) (p33) {};
		\node[colornode=r1] at (-.5,0) (p34) {};
		\node[colornode=b1] at (.5,0) (p35) {};
		\node[colornode=b1] at (1.5,0) (p36) {};
		\node[colornode=b1] at (2.5,0) (p37) {};
		\node[colornode=b1] at (3.5,0) (p38) {};
		
		\node at(-2,-.5) {$v_3^1$};
		\node at(2,-.5) {$v_2^2$};
		\node[colornode=r0] at (-3.5,-1) (p41) {};
		\node[colornode=r0] at (-2.5,-1) (p42) {};
		\node[colornode=r0] at (-1.5,-1) (p43) {};
		\node[colornode=r0] at (-.5,-1) (p44) {};
		\node[colornode=b0] at (.5,-1) (p45) {};
		\node[colornode=b0] at (1.5,-1) (p46) {};
		\node[colornode=b0] at (2.5,-1) (p47) {};
		\node[colornode=b1] at (3.5,-1) (p48) {};
		
		\draw[->] (s) to (p11);
		\draw[->] (s) to (p21);
		\draw[->] (s) to (p31);
		\draw[->] (s) to (p41);
		\draw[->] (p11) to (p12);
		\draw[->] (p21) to (p22);
		\draw[->] (p31) to (p32);
		\draw[->] (p41) to (p42);
		\draw[->] (p12) to (p13);
		\draw[->] (p22) to (p23);
		\draw[->] (p32) to (p33);
		\draw[->] (p42) to (p43);
		\draw[->] (p13) to (p14);
		\draw[->] (p23) to (p24);
		\draw[->] (p33) to (p34);
		\draw[->] (p43) to (p44);
		\draw[->] (p14) to (p15);
		\draw[->] (p24) to (p25);
		\draw[->] (p34) to (p35);
		\draw[->] (p44) to (p45);
		\draw[->] (p15) to (p16);
		\draw[->] (p25) to (p26);
		\draw[->] (p35) to (p36);
		\draw[->] (p45) to (p46);
		\draw[->] (p16) to (p17);
		\draw[->] (p26) to (p27);
		\draw[->] (p36) to (p37);
		\draw[->] (p46) to (p47);
		\draw[->] (p17) to (p18);
		\draw[->] (p27) to (p28);
		\draw[->] (p37) to (p38);
		\draw[->] (p47) to (p48);
		\draw[->] (p18) to (t);
		\draw[->] (p28) to (t);
		\draw[->] (p38) to (t);
		\draw[->] (p48) to (t);
	\end{tikzpicture}
	\end{minipage}
	~
	\begin{minipage}{0.25\columnwidth}
	\begin{tabular}{c|c}
		$\jokercolor$ &\tikz{\node[colornode=yellow,inner sep=3pt] at (0,0) (s) {};}\\
		\hline
		$\colorvar_1$ &\tikz{\node[colornode=r0,inner sep=3pt] at (0,0) (s) {};}\\
		$\colo{1}$ &\tikz{\node[colornode=r1,inner sep=3pt] at (0,0) (s) {};}\\
		\hline
		$\colorvar_2$ &\tikz{\node[colornode=b0,inner sep=3pt] at (0,0) (s) {};}\\
		$\colo{2}$ & \tikz{\node[colornode=b1,inner sep=3pt] at (0,0) (s) {};}\\
	\end{tabular}
	\end{minipage}
	\end{center}
	\caption{The edge-verification gadget for~$V_1$ and~$V_2$ in \cref{fig:mcc} and a legend providing the names for each color.
	Each path represents an edge and both incident vertices are listed above the respective path.}
	\label{fig:edgegadget}
\end{figure}

	Observe that the number of vertices not of color~$\jokercolor$ in each shortest path through all the different gadgets is exactly
	$k \cdot x + \binom{k}{2}y$, which is equal to
	\begin{align*}
		k(k-1) \sum\limits_{\ell=1}^{\digits} k^{\ell-1} + 2 {\textstyle \binom{k}{2}} \sum\limits_{\ell=1}^{\digits} k^{\ell-1}
		&= 2 k (k-1) \sum\limits_{\ell=1}^{\digits} k^{\ell-1}.
	\end{align*}
	Hence, in order for such a path to contain the same amount of vertices of each color, each of the~$2k$ colors different from~$\jokercolor$ have to appear exactly~$(k-1) \cdot \sum_{\ell=1}^{\digits} k^{\ell-1}$ times.
	Moreover, every path passing through all gadgets contains exactly~$K$ vertices of color~$\jokercolor$.
	To complete the construction, we add a path on $(k-1)\cdot \sum_{\ell=1}^{\digits} k^{\ell-1} - K$ vertices of color~$\jokercolor$ to~$H$, call the first vertex in it~$s$, add an arc from its last vertex to~$u_1$, and rename the vertex~$u_K$ to~$t$.
	Observe that~$(k-1)\cdot \sum_{\ell=1}^{\digits} k^{\ell-1} - K \geq 0$ for all~$\digits \ge 2$ and~$k \geq 3$.
	
	It remains to show that the original instance is a \yes-instance if and only if the constructed instance is a \yes-instance and that a running time in~$f(k)\cdot n^{o(\nicefrac{k}{\log k})}$ for some computable function~$f$ would refute the ETH.
	The main idea in there is to analyze the layers separately and for each layer analyze the number of vertices of colors~$\colorvar_i$ and~$\colo{i}$ in the respective layer in all (vertex-selection and edge-verification) gadgets combined.
	
       For the forward direction, suppose that~$G$ contains a clique~${C = \{v_{a_1}^1,v_{a_2}^2,\ldots,v_{a_k}^k\}}$ with~$k$ vertices.
       Consider the path~$P$ from~$s$ to~$t$ that contains all vertices of color~$\jokercolor$, for each~$V_i$ the path representing vertex~$v_{a_i}^i$ in the vertex-selection gadget for~$V_i$, and for each pair~$V_i \neq V_{j}$ the path representing~$\{v_{a_i}^i,v_{a_{j}}^{j}\}$ in the edge-verification gadget for~$V_i$ and~$V_{j}$.
       Note that all these edges exist since~$C$ is a clique.
       We claim that~$P$ contains exactly~${(k-1) \cdot \sum_{\ell=1}^{\digits} k^{\ell-1}}$~vertices of each color.
       First, this is true for the color~$\jokercolor$.
       Now consider any color~$\colorvar_i$ or~$\colo{i}$ for some~$i\in[k]$.
       Since~$\colorvar_i$ and~$\colo{i}$ only appears in the vertex-selection gadget for~$V_i$ and in edge-verification gadgets for~$V_i$ and~$V_j$ for each~$j\neq i$, we will ignore all other gadgets.
       We will show that~$P$ contains exactly~$(k-1)\cdot k^{\ell-1}$~vertices of colors~$\colorvar_i$ and~$\colo{i}$ in all~$\ell$-th levels of the considered gadgets combined.
       To this end, we make a case distinction whether the~$\ell$-th least significant digit~$d_\ell$ in the binary encoding of~$a_i$ is~$0$ or~$1$.
       If~$d_\ell = 0$, then the path~$P_v$ representing~$v^i_{a_i}\in C$ in the vertex-selection gadget for~$V_i$ contains~$(k-1)\cdot k^{\ell-1}$ vertices of color~$\colorvar_i$ and no vertices of color~$\colo{i}$ in the~$\ell$-th level and a path~$P_e$ representing an edge containing~$v^i_{a_i}$ in an edge-verification gadget contains by construction $k^{\ell-1}$ vertices of color~$\colorvar_i$ and no vertices of color~$\colo{i}$ in the~$\ell$-th level.
       Analogously, if~$d_\ell = 1$, then the path~$P_v$ representing~$v^i_{a_i}\in C$ in the vertex-selection gadget for~$V_i$ contains no vertices of color~$\colorvar_i$ and~$(k-1)\cdot k^{\ell-1}$ vertices of color~$\colo{i}$ in the~$\ell$-th level and a path~$P_e$ representing an edge containing~$v^i_{a_i}$ in an edge-verification gadget contains no vertices of color~$\colorvar_i$ and~$k^{\ell-1}$ vertices of color~$\colo{i}$ in the~$\ell$-th level.
       Since there are exactly~$k-1$ edge-verification gadgets for~$V_i$ and some~$V_{j}$ with~$j\neq i$, it holds that there are exactly~$(k-1) \cdot \sum_{\ell=1}^{\digits} k^{\ell-1}$ vertices of color~$\colorvar_i$ and~$\colo{i}$ in~$P$.
       Moreover, since we chose~$i$ arbitrarily, it holds that $P$~contains exactly~$(k-1) \cdot \sum_{\ell=1}^{\digits} k^{\ell-1}$ vertices of each color.
       Thus, the resulting instance of \bsp{} is a \yes-instance.
       
       For the backward direction, assume that there is a \fair{} shortest \stpath~$P$ in~$H$, that is, $P$ contains exactly~${(k-1) \cdot \sum_{\ell=1}^{\digits} k^{\ell-1}}$~vertices of each color.
       Consider the set~$C' = \{v_{b_1}^1,v_{b_2}^2,\ldots,v_{b_k}^k\}$ of vertices corresponding to those paths of the vertex-selection gadgets that are contained in~$P$.
       We will show that~$C'$ forms a clique in~$G$.
       To this end, consider any vertex~$v_{b_i}^i \in C'$.
       We claim that~$P$~contains~$k-1$~paths of the edge-verification gadgets representing edges incident to~$v_{b_i}^i$.
       Note that there are~$\binom{k}{2}$ edge-verification gadgets and we chose~$v_{b_i}^i \in C'$ arbitrarily.
       Hence, if the claim holds, then~$P$ contains~$\binom{k}{2}$ paths representing edges between~$k$ vertices.
       This implies that~$C'$ is a (multicolored) clique of size~$k$, concluding the proof.
       Assume towards a contradiction that the claim does not hold, that is, there are subpaths of~$P$ through edge-verification gadgets that encode edges with an endpoint in~$V_i$ other than~$v_{b_i}^i$.
       Let~$F$ be the set of these endpoints and let~$\ell$ be the smallest number such that the~$\ell$-th least significant digit in the binary encoding of~$b_i$ differs from the~$\ell$-th least significant digit in the binary encoding of some~$c$ with~$v_{c}^i \in F$.
       We assume that the~$\ell$-th least significant digit of~$b_i$ is~$1$.
       The case for the~$\ell$-th least significant digit being~$0$ is completely analogous with colors~$\colorvar_{i}$ and~$\colo{i}$ swapped.
       Consider the number of vertices of color~$\colorvar_i$ in~$P$.
       By the minimality of~$\ell$, it holds for all~$\ell' < \ell$ that the~$\ell'$-th least significant digit of~$c$ and~$b_i$ are the same for all~$v_{j}^i \in F$.
       Hence, $P$ contains~$(k-1)\cdot k^{\ell'}$ vertices of color~$\colorvar_i$ in the~$\ell'$-th level of all (vertex-selection and edge-verification) gadgets combined.
       It follows that~$P$ contains~$(k-1) \cdot \sum_{d=1}^{\ell-1} k^{d-1}$ vertices of color~$\colorvar_i$ in the first $\ell-1$~levels in all gadgets.
       Since the target number for each color is~$(k-1) \cdot \sum_{d=1}^{\digits} k^{d-1}$ and each level~$\ell' > \ell$ contains each color a multiple of~$k^\ell$ times, it follows that the number of vertices of color~$\colorvar_i$ modulo~$k^\ell$ in the~$\ell$-th level in all gadgets is~\mbox{$(k-1)\cdot k^{\ell-1}$}.
       By our assumption, there are no vertices of color~$\colorvar_i$ in the~$\ell$-th level of the vertex-selection gadget for~$V_i$.
       Since each edge-verification gadget contains at most $k^{\ell-1}$~vertices of color~$\colorvar_i$ in the~$\ell$-th level, it follows that each edge-verification gadget for~$V_i$ and~$V_{j}$ with~$j \neq i$ contains exactly~$k^{\ell-1}$ vertices of color~$\colorvar_i$ in the~$\ell$-th level.
       This contradicts the assumption that the~\mbox{$\ell$-th} least significant digit in the binary encoding of~$b_i$ differs from the~$\ell$-th least significant digit in the binary encoding of some~$c$ with~$v_{c}^i \in F$.

       In order to show the running-time lower bound, assume that there are computable functions~$f$ and~$g$ with ${g(x) \in o(\nicefrac{x}{\log x})}$ such that \bsp{} parameterized by the number~$c$ of colors can be solved in~$f(c) \cdot n^{g(c)}$ time.
       Let~$f'(x) \defeq f(2x+1)$, $g^*(x) \defeq g(x) \cdot \log(x)$, and~${g'(x) \defeq g^*(2x+1)}$.
       As~$g^*(x) \in o(x)$ we have~$g'(x) \in o(x)$.
       It remains to show that \mcc{} can be solved in~$\bigO(f'(k) \cdot n^{g'(k)})$ time.
       Since
       \[\textstyle\sum_{\ell=1}^{\digits} k^{\ell-1} = \sum_{\ell=1}^{\lceil \log \eta\rceil} k^{\ell-1} = \sum_{\ell=0}^{\lfloor \log \eta\rfloor} k^{\ell} \leq k^{1+\log \eta},\]
       the number~$n'$ of vertices in the constructed instance as well as the running time for the reduction are in~${k^{\bigO(\log n)} = 2^{\bigO(\log n \cdot \log k)} = n^{\bigO(\log k)}}$.
       Moreover, observe that $c = 2k +1$.
       Thus, first computing the described reduction and then solving the resulting instance of \bsp{} in~$f(c) \cdot (n')^{g(c)}$ time results in an overall running time of~$n^{\bigO(\log k)} + f(2k+1) \cdot (n')^{g(2k+1)} \subseteq f'(k) \cdot n^{g'(k)}$.
\end{proof}

We remark that there is still a gap between the running time upper bound of~$n^{\bigO(c)}$ (\cref{prop:cXP}) and the lower bound of~$n^{o(\nicefrac{c}{\log c})}$ (\cref{thm:eth}).

\section{The Parameter Path Length}
\label{sec:length}

This section is devoted to proving that \splu{} is fixed-parameter tractable with respect to the length~$\ell$ of the path.
More specifically, the running time of our algorithm depends on the number~$k$ of vertices that are visited by the sought solution path.
Recall that we assume that each edge length is at least one; thus we have $k \le \ell + 1$.
This is a reasonable assumption as otherwise the problem is NP-hard even if~$\ell=0$ (\cref{prop:0-arc-np}).

Throughout this section, we will make use of the fact that we can model the balance constraints as a matroid.
Then, making use of the representative-families framework by \citet{fomin2016representative}, we obtain the algorithm.

We model our lower-bound and upper-bound constraints as a matroid~$M = (V, \III)$.
Its bases are exactly those sets~$B \subseteq V$ with~$\abs{B} = k$ and ${\low_i \le \abs{B \cap \chi^i} \le \upp_i}$.
Thus its family of independent sets is
\[ \III = \{X \subseteq V \mid \abs{X \cap \chi^i} \le \upp_i \text{ and } \abs{X}+g_X \le k,\, i \in \oneto{c}\}. \]
Here, $g_X = \sum_{i=1}^c \max\{0, \low_i - \abs{X \cap \chi^i}\}$, and~$\abs{X} + g_X \le k$ ensures that there is still enough ``room'' to add vertices of colors whose lower bounds are not yet met.

By definition of~$M$, any path~$P$ on exactly~$k$ vertices fulfills the lower and upper bounds if and only if~$V(P)$ is a basis of~$M$.
We call a vertex set~$X \subseteq V$ \emph{$M$-extendable\footnote{We drop the prefix if there is no ambiguity.}} if~$X$ is independent in~$M$, that is, $X \in \III$.
Similarly, a path~$P$ is $M$-extendable if~$V(P) \in \III$.

Define for each~$v \in V$ and~$p \in \oneto{k}$ the family
\begin{equation*}
	\begin{aligned}
		\PPP^{p}_{v} \defeq \big\{& X \subseteq V \,\,\big\lvert\,\, \abs{X} = p,\, X \in \III,\\ &\text{and there is an \abpath{s}{v} } P \text{ with } V(P) = X \big\}.
	\end{aligned}
\end{equation*}

We now may find an \stpath{} fulfilling the lower and upper bounds with a simple Dijkstra-like dynamic programming approach.
This would not yield the promised running time, as each set~$\PPP^p_v$ contains up to~$\binom{n}{p}$ subsets.
The crucial observation here is that~$\PPP^p_{v}$ stores \emph{too much}:
Consider a set~$Y \subseteq V$.
If there is a set~$X \subseteq V$ with~$X \cap Y = \emptyset$ such that~$X \cup Y$ is a basis of~$M$, then we say that~$Y$ \emph{extends} $X$, and that~$X$ is \emph{extendable by}~$Y$.
Now, for each set~$Y$, we only need to remember \emph{one} set~$X$ that is extendable by~$Y$.
Hence, if~$Y$ extends multiple sets in our family, we only need to remember one of them as a \emph{representative} for the remaining sets.
We consider a weighted variant of this concept, defined as follows.
\begin{definition}[min-$q$-representative family]
	\label{def:representative}
	Given a matroid~$M = (U, \III)$, a family~$\SSS$ of subsets of~$U$ of size~$p$, and a weight function~$\rho \colon \SSS \to \N$,
	a subfamily~$\widehat\SSS \subseteq \SSS$ is a \emph{min-$q$-representative} for~$\SSS$ ($\widehat\SSS \repr^q \SSS$) if the following holds:
	for every set~$Y \subseteq E$ of size at most~$q$,
	if there is a set~$X \in \SSS$ disjoint from~$Y$ with~$X \cup Y \in \III$, then there is a set~$\widehat X \in \widehat\SSS$ disjoint from~$Y$ with
	\begin{inparaenum}[(1)]
		\item $\widehat X \cup Y \in \III$ and
		\item $\rho(\widehat X) \le \rho(X)$.
	\end{inparaenum}
\end{definition}

We remark that representative families are transitive: If~$\AAA \repr^q \BBB$ and~$\BBB \repr^q \CCC$, then~$\AAA \repr^q \CCC$.
Apart from proving their transitivity, \citet{fomin2016representative} presented an efficient algorithm to compute such representative families.
We denote by~$\omega < 2.373$ the matrix multiplication constant \cite{matrixmultiplication}.

\begin{theorem}[\citet{fomin2016representative}]
	\label{thm:repr}
	Let~$M = (U, \III)$ be a linear matroid of rank~$p+q=k$ given together with its representation matrix~$A$ over a field~$\FF$.
	Let~$\AAA = \{A_1, \dots, A_t\}$ be a family of independent sets of size~$p$ and let~$w \colon \AAA \to \N_0$ be a weight function.
	Then, a min-$q$-representative family~$\widehat\AAA$ for~$\AAA$ with at most~$\binom{p+q}{p}$ sets can be found in
	\[
		\bigO\big(\textstyle\binom{p+q}{p}tp^{\omega} + t\binom{p+q}{q}^{\omega-1}\big)
	\]
	operations over~$\FF$.
\end{theorem}

To be able to make use of this theorem with our matroid~$M$, we prove the following.

\begin{restatable}{lemma}{MIsGammoid}
	\label{lem:m-linear}
	$M$ is a gammoid of rank~$k$.
\end{restatable}
\begin{proof}
	The directed graph~$H$ corresponding to~$M$ consists of vertex sets~$S_i$ and~$U_i$ for each~$i \in \oneto{c}$ and of sets~$S^*$ and~$V$.
	For each~$i \in \oneto{c}$, $\abs{S_i} = \low_i$ and $\abs{U_i} = \upp_i - \low_i$, and~$\abs{S^*} = k - \sum_{i \in \oneto{c}} \low_i$.
	The graph~$G$ contains all edges between the following vertex subsets for all~$i \in \oneto{c}$: from~$S^*$~to~$U_i$, from~$S_i$~to~$V \cap \chi^i$, and from~$U_i$~to~$V \cap \chi^i$.
	We claim that a set~$X \subseteq V$ is independent in $M$ if and only if $X$ is linked to $S \coloneqq S^* \cup \bigcup_{i \in \oneto{c}} S_i$.

	Suppose~$X$ is independent in $M$.
	Then, we route ${\min\{\abs{X \cap \chi^i}, \low_i\}}$~paths for each color~$i \in \oneto{c}$ starting in~$S_i$ and ending in~$X \cap \chi^i$.
	The remaining (but at most~$\upp_i - \low_i = \abs{U_i}$) elements in~$X$ with color~$i$ are represented by paths starting in~$S^*$, visiting a vertex in~$U_i$, and ending in~$X \cap \chi^i$.
	Since~$\abs{X} \le \ell$, the set~$S^*$ is sufficiently large to provide the required number of vertex-disjoint paths for each color~$i$.

	Now suppose that~$X$ is linked to~$S$, that is, there exists a family~$\FFF$ of~$\abs{X}$ vertex-disjoint paths starting in~$S$ and ending in~$X$.
	For each color~$i$ let~$\gamma_i$ be the number of paths in~$\FFF$ from a vertex in~$S_i$ to a vertex in~$X \cap \chi^i$,
	and let~$\delta_i$ be the number of paths in~$\FFF$ from a vertex in~$S^*$ via a vertex in~$U_i$ to a vertex in~$X \cap \chi^i$.
	Then $\abs{X \cap \chi^i} \le \gamma_i + \delta_i \le \beta_i$.
	As~$\FFF$ can be extended by~$\low_i - \gamma_i$ paths starting in~$S_i$ and ending in a vertex in~$(V \setminus X) \cap \chi^i$ and by ${\ell - \sum_{i\in\oneto{c}} \low_i - \sum_{i=1}^\ell \delta_i}$~paths starting in~$S^*$ and ending in a vertex in~$V \setminus X$ of arbitrary color,
	we have that there is a set~$B$ of $\ell$ vertex-disjoint paths from~$S$ to~$V$ such that~$X \subseteq B$.
	Clearly, $B$ is a basis of $M$ and~$X$ is therefore independent in~$M$.
\end{proof}

We obtain a linear representation in polynomial time from \cref{thm:gammoid-repr},
which may be errorneous with low probability.
Indeed, the representation is actually the representation of \emph{another} matroid~$\widetilde M = (V, \widetilde\III)$ such that for each~$X \subseteq V$ we have~$X \notin \widetilde\III$ if~$X \notin \III$ and we have~$\Pr[X\in\widetilde\III] \ge (1-\eps)$ if~$X \in \III$.
As we will compute min-$q$-representatives corresponding to~$\widetilde M$, it is convenient to introduce a corresponding adaptation for~$\PPP^p_v$.
For each~$v \in V$ and~$p \in \oneto{k}$, let
\begin{equation}
	\label{eq:random-family}
	\begin{aligned}
		\widetilde\PPP^{p}_{v} \defeq \big\{& X \subseteq V \,\,\big\lvert\,\, \abs{X} = p,\, X \in \widetilde\III\,\\ &\text{and there is an \abpath{s}{v} } P \text{ with } V(P) = X \big\}.
	\end{aligned}
\end{equation}

Our algorithm for \splu{} then builds representative families of~$\widetilde\PPP^p_v$ inductively.
For each family~$\PPP^p_v$, we also keep a weight function~$\rho^p_v$ which is intended to store for each set~$X \in \PPP^p_v$ the weight of the shortest~\abpath{s}{v} with vertex set~$X$.

\begin{alg}\label{alg:fpt-l}
	Set $\widehat\PPP^1_s \defeq \{\{s\}\}$,~$\rho(\{s\}) \defeq 0$, $\widehat\PPP^1_v \defeq \emptyset$ for each~$v \in V \setminus \{s\}$, and~$\rho(X) \defeq \infty$ for each~${\{s\} \ne X \subseteq V}$.
	For~$p = 2, \dots, k$, compute for each~$v \in V$
	\begin{equation}
		\label{eq:path-repr-recurrence}
		\NNN^{p}_{v} = \bigcup_{(u, v) \in E}
		\left\{X \cup \{v\} \Big\vert\ \begin{aligned}
		& v \notin X,\, X \in \widehat\PPP^{p-1}_{u},\\ &\text{ and } X \cup \{v\} \in \widetilde\III
		\end{aligned}\right\},
	\end{equation}
	and for each set~$X' = X \cup \{v\}$ to add to~$\NNN^{p}_{v}$, set~$\rho^p_v(X') \defeq \min\{ \rho^p_v(X'), \rho^{p-1}_u(X) + w((u,v))\}$.
	Afterwards, compute~${\widehat \PPP^p_v \repr^{k-p} \NNN^p_v}$ for each~${v \in V}$.
	Finally, return \yes{} if there is an~$X^* \in \widehat\PPP^k_t$ with~$\rho^k_t(X^*) \le \ell$.
\end{alg}

Let us show that \cref{alg:fpt-l} is correct.

\begin{lemma}
	\label{lem:path-repr}
	For each~$v \in V$ and~$p \in \oneto{k}$, $\widehat \PPP^{p}_v$ contains at most~$\binom{k}{p}$ sets and is a~\mbox{min-$(k-p)$}-representative for~$\widetilde \PPP^p_v$.
	For each~${X \in \widehat \PPP^{p}_v}$, the weight~$\rho^p_v(X)$ equals the length of a shortest $\widetilde M$-extendable \abpath{s}{v} with vertex set $X$.
\end{lemma}
\begin{proof}
	Our proof is by induction.
	The statement is trivially correct for~$p=1$.
	So assume that for some~$p > 1$, $\widehat\PPP^{p-1}_v \repr^{k-(p-1)} \widetilde\PPP^{p-1}_v$, $\abs{\widehat\PPP^{p-1}_v} \le \binom{k}{p}$, and for each~$X \in \widehat\PPP^{p-1}_v$,
	the weight~$\rho^{p-1}_v(X)$ equals the length of the shortest~$\widetilde M$-extendable \abpath{s}{v} on~$p-1$ vertices.

	Let~$Y \subseteq V$ be a set of size at most~$k-p$ and let~$v \in V$.
	Suppose that there exists a set~$X \in \widetilde\PPP^p_v$ such that~$Y$ $\widetilde M$-extends~$X$.
	Let~$P$ be the path corresponding to~$X$, and let~$w(P) = \rho^p_v(X)$.
	Consider the subpath~$P'$ of~$P$ going from~$s$ to the predecessor~$u$ of~$v$.
	Then~$P'$ contains~$p-1$ vertices and is of length~$w(P') = w(P) - w((u, v)) \ge \rho^{p-1}_u(X')$.
	Let~$X' \defeq V(P)$ and
	let~$Y' \defeq Y \cup \{v\}$.
	By the hereditary property of~$\widetilde M$,
	the path~$P'$ is also~$\widetilde M$-extendable.
	Further, as~$X \cup Y \in \widetilde I$ and~$X \cap Y = \emptyset$,
	we also have that~$Y'$ $\widetilde M$-extends~$X'$ (note that~$X' \cap Y' = \emptyset$).
	Hence, by the induction hypothesis, there a set~$\widehat X' \in \widehat\PPP^{p-1}_u$
	with $\rho^{p-1}_u(\widehat X') \le \rho^{p-1}_u(X')$
	such that~$Y$ $\widetilde M$-extends~$X$.
	Let~$\widehat P'$ be an~\mbox{\abpath{s}{u}} of length~$\rho^{p-1}_u(\widehat X')$ corresponding to~$\widehat X'$
	and let~$\widehat X \defeq \widehat X' \cup \{v\}$.
	By Equation~\eqref{eq:path-repr-recurrence} we have~$\widehat X \in \NNN^p_v$
	and~$\rho^p_v(\widehat X) \le \rho^{p-1}_u(\widehat X') + w((u, v)) \le \rho^p_v(X)$.
	Thus, $\NNN^p_v \repr^{k-p} \widetilde\PPP^p_v$.
	Now, as~$\widehat\PPP^p_v \repr^{k-p} \NNN^p_v$, we obtain~$\widehat\PPP^p_v \repr^{k-p} \widetilde\PPP^p_v$ by transitivity of representatives \cite{fomin2016representative}.
	The upper bound on~$|\widehat\PPP^p_v|$ follows from \cref{thm:repr}.
\end{proof}

We next analyze the running time of \cref{alg:fpt-l}.

\begin{restatable}{lemma}{FptLTime}
	\label{lem:time-fpt-l}
	\cref{alg:fpt-l} runs in~$2^{\omega k} n^{\bigO(1)}\log(1/\eps)$ time,
	where~$\eps$ is the error probability used to compute~$\widetilde M$.
\end{restatable}
\begin{proof}
	Computing~$\NNN^p_v$ for some~$p \in \oneto{k}$ and~$v \in V$ takes
	$ %
		\bigO(\textstyle n \cdot |\widehat \PPP^{p-1}_u|) \subseteq \bigO(\textstyle n \binom{k}{p-1})
	$ %
	time.
	This is also a bound on the size of~$\widetilde\NNN^p_v$.
	As the entries in the representation of~$\widetilde M$ have bit-length~$\bigO(\min \{n, k \log n\} + \log(1/\eps) + \log(k + n)) \subseteq \bigO(n \log n + \log (1/\eps))$, the time to compute~$\widehat \PPP^p_v \repr^{k-p} \NNN^p_v$ is
	\begin{align*}
		\bigO\Big(\textstyle n \abs{\NNN^p_v} \cdot \big( \binom{k}{p} p^\omega + \binom{k}{p}^{\omega-1} \big) \cdot (n \log n + \log(1/\eps)) \Big) \subseteq \bigO\big(\textstyle \binom{k}{p} \cdot n^{\bigO(1)} \cdot \log(1/\eps)\big).
	\end{align*}
	Note that this dominates the time to compute~$\NNN^p_v$.
	Hence, the total running time of our algorithm is upper-bounded by
	\begin{align*}
		&\sum_{p\in\oneto{k}} \bigO \big({\textstyle \binom{k}{p}^\omega} n^{\bigO(1)} \log(1/\eps)\big) \subseteq \bigO(2^{\omega k} n^{\bigO(1)} \log(1/\eps)). \qedhere
	\end{align*}
\end{proof}

Our main theorem now follows from \cref{lem:path-repr,lem:time-fpt-l}.

\begin{restatable}{theorem}{FptLThm}
	\label{thm:fpt-l}
	Let~$\eps > 0$.
	Then there is a randomized algorithm without false positives and success probability at least~$(1-\eps)$ that solves \splu{}
	in~$5.181^{\ell} \log (1/\eps) n^{\bigO(1)}$ time.
\end{restatable}
\begin{proof}
	We first guess the number~$k \le \ell+1$ of vertices in our path.
	Then, we compute a representation of~$M$ with error probability~$\eps$ in polynomial time using \cref{thm:gammoid-repr}.
	Afterwards, we run \cref{alg:fpt-l}.
	Indeed, this yields the promised result:
	If our instance is a \no-instance, then for every \stpath{} $P$ on~$k$ vertices with length at most~$\ell$,
	we have~$V(P) \notin \III$, which in turn implies that~$V(P) \notin \widetilde\III$ by \cref{thm:gammoid-repr}.
	Hence, we have~$\widetilde \PPP^k_t = \emptyset$, which, due to \cref{lem:time-fpt-l}, implies that~$\widehat\PPP^k_t = \emptyset$;
	thus the algorithm returns \no.
	Conversely, if our instance is a \yes-instance,
	then there exists an \stpath{} $P$ of length~$\ell$ on~$k$ vertices such that~$V(P) \in \III$;
	thus~$\Pr[V(P) \in \widetilde\III] \ge 1-\eps$.
	Then~$\Pr[\widetilde\PPP^k_t] \ge 1-\eps$, so~$\Pr[\widehat\PPP^k_t] \ge 1-\eps$;
	therefore the probability that the algorithm returns \yes{} is at least~$1-\eps$.
	Hence, the algorithm produces no false positives.
	The overall running time then comprises of the time to compute the representation for~$M$ and the running time of the algorithm.
	As the graph corresponding to the gammoid has size~$\bigO(k + \sum_{i\in\oneto{c}} \upp_i + n) \subseteq \bigO(n)$, and as~$2^\omega < 5.181$, our algorithm terminates within the claimed running time.
\end{proof}

\section{Conclusion}
\label{sec:concl}
Our study provides insight into a very general fairness model injected into the problem of finding short paths in graphs.
We indicate that even very simple cases of the considered problem turn out to be computationally hard,
but, if the sought path is sufficiently short, then the problem can be efficiently solved.
For parameterizing \bsp{} by the number~$c$ of colors,
our results leave a gap between the running time upper bound of~$n^{\bigO(c)}$ and the lower bound of~$n^{o(\nicefrac{c}{\log c})}$.
In a follow-up work, we studied the parameterized complexity of \bsp{} with respect to structural parameterizations \cite{bentert2022structural}.

Our fairness model generalizes upon many established variants.
Another benefit of this fairness concept is that it can be modeled by a matroid.
We make use of this to obtain an efficient (randomized) parameterized algorithm.
The pervasiveness of matroids lets us believe that our fairness concept may be easily applicable also to other problems.
Indeed, it can also be applied to a problem studied by \citet{chierichetti2017clustering}, which we will call \prob{Balanced Matroid}:
Given a matroid~$Q = (U, \III_Q)$ of rank~$\ell$ for which one can decide in polynomial time whether a set~$X$ is independent in~$Q$ or not, a coloring~$\chi \colon U \to \oneto{c}$, and two rationals~$0 \le \alpha \le \beta \le 1$,
the task is to find the independent set~$X$ of largest cardinality such that for each color~$i$, the fraction of elements of color~$i$ in~$X$ lies between~$\alpha$ and~$\beta$.
They prove that if~$c$ is constant, then this problem can be solved in polynomial time~\mbox{\cite[Lem.~3]{CKLV19}}.
Indeed, this problem is polynomial-time solvable even if~$c$ is not a constant.
To this end, one first guesses the cardinality~$p \le \ell$,
then one uses the matroid~$M$ devised in the previous section, choosing~$\low_i \coloneqq \lceil \alpha \cdot p \rceil$ and~$\upp_i \coloneqq \lfloor \beta \cdot p \rfloor$ for each color~$i$.
Using the matroid intersection algorithm \cite{edmonds1970submodular}, one can compute in polynomial time the maximum-cardinality set which is independent in both~$Q$ and~$M$.
If the set has cardinality~$p$, then the guess was correct.

\begin{corollary}
	\prob{Balanced Matroid} is polynomial-time solvable.
\end{corollary}

There are also some variants which do not fit into our model.
An intriguing model variation is to enforce fairness not only to the path as a whole, but also to each (sufficiently long) subpath of the solution.
This problem seems computationally easier to tackle as the constraints are more local.
A natural extension that feels harder to tackle is the case in which a vertex may hold multiple colors (e.g., a place for both biking and hiking).

\paragraph{Acknowledgment.}
We thank an anonymous reviewer for an observation that lead to the ETH-based lower bound and a significant simplification of the proof of \cref{thm:cW1}.

\sloppy
\printbibliography

\end{document}